\documentclass[11pt, a4paper]{article}

\usepackage[margin=1in]{geometry}
\usepackage[T1]{fontenc}

\usepackage{amsmath, amssymb, amsthm}
\usepackage{latexsym, xspace}
\usepackage{enumitem}

\newtheorem{theorem}{Theorem}
\newtheorem{lemma}[theorem]{Lemma}

\newcommand{\es}{\varnothing}
\newcommand{\hn}{\hspace{-1.5ex}}

\newcommand{\sm}{\setminus}
\newcommand{\cov}{\mathbin{\prec\hspace{-1ex}\cdot}}
\newcommand{\scov}{\mathbin{\prec\hspace{-.25ex}\cdot}}
\newcommand{\preq}{\preccurlyeq}

\title{Counting Independent Sets in Cocomparability Graphs}
\author{Martin Dyer \and Haiko M\"{u}ller}
\date{School of Computing, University of Leeds, Leeds LS2~9JT, UK \\[.5\baselineskip]
  \texttt{\{M.E.Dyer|H.Muller\}@leeds.ac.uk}
}

\begin{document}

\maketitle

\begin{abstract}
We show that the number of independent sets in cocomparability graphs can be counted in linear time,
as can counting cliques in comparability graphs. By contrast, counting cliques in cocomparabilty graphs
and counting independent sets in comparabilty graphs are \#P-complete. We extend these results to counting
maximal cliques and independent sets. We also consider the fixed-parameter versions of counting
cliques and independent sets of given size $k$. Finally, we combine the results to show that both
counting cliques and independent sets in permutation graphs are in linear time.
\end{abstract}

\section{Introduction}
Counting independent sets in graphs is known to be \#P-complete in general~\cite{PrBa83},
even for several restricted cases \cite{Vadhan}. Even approximate counting is NP-hard,
and unlikely to be in polynomial time for bipartite graphs~\cite{DGGJ03}. However,
smaller graph classes may permit polynomial time algorithms. For example, \cite{bip perm}
gives a linear time algorithm for counting independent sets in bipartite permutation graphs.
Here ``linear time'' means computable with a number of arithmetic operations and comparisons
linear in the size of the graph. Similar results are known for \emph{chordal graphs}~\cite{chordal}
and \emph{tree-convex graphs} \cite{tree convex}. For information on these and other
graph classes, see \cite{survey}.

Counting cliques in a graph is equivalent to counting independent sets in the complementary graph.
Thus, for general graphs, these two problems have equivalent complexity. However, for graph classes
the two problems may have very different complexity.

Here we consider the problems of counting independent sets and cliques in
three classes of graphs: \emph{cocomparability graphs}, \emph{comparability graphs}
and \emph{permutation graphs}. See section~\ref{sec:prelim} below for definitions.
In each case we show that either the counting problem
can be accomplished in linear time, or else counting is \#P-complete.

In section~\ref{sec:cocomp} we present a simple linear time algorithm for counting independent sets
in cocomparability graphs, and show that counting cliques is \#P-complete.
We extend this to counting maximal independent sets, and independent sets of a given size,
and maximal independent sets of a given size.
In section~\ref{sec:comp}, we modify this to give linear time algorithms for counting cliques
in comparability graphs, and show that counting independent sets is \#P-complete.

Together, these results imply simple linear time algorithms for counting both cliques and independent sets
in permutation graphs, as we discuss in section~\ref{sec:perm}. Thus our results give strong generalisations
of the algorithm of \cite{bip perm}.

\section{Preliminaries and notation}\label{sec:prelim}
For integers $n$ and $n'$ with $n \le n'$ the set $\{n, n+1, \dots, n'\}$
of consecutive integers will be abbreviated as $[n,n']$.

Let $P=(V,\prec)$ be a finite partial order, with $|V|=n$.
The \emph{cover relation} $\cov$ of $\prec$ is its transitive reduction,
defined by $u \cov v$ if $u \prec v$ and there is no $w \in V$ with $u \prec w$
and $w \prec v$. By $u \preq v$ we mean $u \prec v$ or $u=v$.

A \emph{chain} in $(V,\prec)$ is a set $S \subseteq V$ of pairwise comparable
elements. That is, the restriction of $\prec$ to $S$ is a linear order. A chain
$S = \{v_i \mid i \in [1,k]\}$ with $v_{i-1} \prec v_i$ for all $i \in [2,k]$
is \emph{tight} if $v_{i-1} \cov v_i$ holds for all $i \in [2,k]$.

A \emph{linear extension} of $(V,\prec)$ is a permutation $v_1,v_1,\ldots,v_n$ of $V$ such that
$v_i \prec v_j$ implies $i<j$.

Let $G=(V,E)$ be a graph with vertex set $V$,
with $|V|=n$ and $|E|=m$. Here all graphs are finite, undirected and simple,
unless stated otherwise. We will denote the complementary graph by $\bar{G}=(V,\bar{E})$,
where $\bar{E}=\{\{u,v\} \mid u,v\in V, \{u,v\}\notin E, u \ne v\}$. Let $\bar{m}=|\bar{E}|$, that is $\bar{m}=\binom{n}{2}-m$,
and $m^*=\min\{m,\bar{m}\}$.

A \emph{clique} in a graph is a set of pairwise adjacent vertices.
A set of pairwise non-adjacent vertices is \emph{independent}.
Therefore a clique of $G$ is an independent set in the complement of $G$ and vice versa.
A \emph{maximal} clique or independent set is maximal with respect to set
inclusion. By contrast, \emph{maximum} refers to the largest cardinality.

A partial order $P=(V,\prec)$ can be considered as digraph $(V,A)$
with vertex set $V$ and arcs $(x,y) \in A$ for all pairs such that $x\prec y$.
The comparability graph for $P$ is the underlying undirected
graph $(V,E)$ of $(V,A)$. That is, $E = \{\{x,y\} \mid x \prec y\}$. Then a graph
$G=(V,E)$ is a \emph{comparability graph} if there is a partial order $P$ on the set
$V$ such $G$ is the comparability graph of $P$. Given a comparabity graph $G$,
a corresponding partial order $P$ and a linear extension can be computed in
$O(m+n)$ (linear) time \cite{McCSpi94}. Note, however, that this algorithm does
not certify that $G$ is a comparability graph. Recognition is currently only known
in $O(n+m\log n)$ time. We will sidestep these issues by assuming the input comparability
graph is given with a transitive orientation $P$ and a linear extension of this.

The complements of comparability graphs are called \emph{cocomparability graphs}, so
these are the incomparability graphs of partially ordered sets. Again, recognition
is not known in linear time, but a linear time algorithm can compute $P$ and a
linear extension for the complementary comparability graph \cite{KM,McCSpi94}. Again,

A graph is a \emph{permutation graph} if there is a permutation $v_1,v_2,\ldots,v_n$ of $V=[1,n]$
such that $\{v_i,v_j\}\in E$ if and only if $i<j$ implies $v_i>v_j$. It is not difficult to show that $G$ is a
permutation graph if and only if it is both a comparability graph and a cocomparability graph.
Permutation graphs can be recognised, and a permutation ordering obtained, in linear time~\cite{KMMS},
so we need no assumptions on the input graph.

These are well-established classes of graphs, and are all subclasses of \emph{perfect graphs}.
Important subclasses of comparability graphs are \emph{bipartite graphs} and \emph{cographs}.
\emph{Interval graphs} are an important subclass of cocomparability graphs
(they are graphs that are both chordal and cocomparability).
Permutation graphs are an important subclass of both comparability and cocomparability  graphs.

\section{Cocomparability graphs}\label{sec:cocomp}

Let $G=(V,E)$ be a cocomparability graph and let $\prec$ be a partial
order on $V$, with linear extension $v_1,v_2,\ldots,v_n$.
We will extend the poset $(V,\prec)$ by a unique minimal
element $\bot \notin V$ and a unique maximal element
$\top \notin V\cup\{\bot\}$. That is, we define $\bot \prec v$ and $v \prec \top$ for
all $v \in V$, and (if not enforced by transitivity) $\bot \prec \top$.
Let $S^+ = S \cup \{\bot,\top\}$, for any $S\subseteq V$. Then we
denote the extended partial order by $(V^+,\prec)$, If
$\bot,v_1,v_2,\ldots,v_n,\top$ is  linear extension of $(V^+,\prec)$,
and $S\subseteq V$, we write $\max(S)$ for $v_i$ such that $i=\max\{j:v_j\in S\}$,

\begin{lemma} \label{l:chain}
  Let $G=(V,E)$ be a cocomparability graph which is the incomparability graph of a poset
  $(V,\prec)$, and let $(V^+,\prec)$ be the above extension. Then
  \begin{enumerate}[label=(\arabic*),topsep=0pt,itemsep=0pt]
  \item \label{is=c} a set $S \subseteq V$ is independent in $G$
        if and only if $S^+$ is a chain of $(V^+,\prec)$;
  \item \label{mis=tc} a set $S \subseteq V$ is maximally independent in $G$
        if and only if $S^+$ is a tight chain of $(V^+,\prec)$.
  \end{enumerate}
\end{lemma}

\begin{proof}
  Two vertices in $V$ are $\prec$-comparable if and only if they non-adjacent in
  $G$. The extra elements $\bot$ and $\top$ are comparable to all vertices
  in $V$. Together these imply property~\ref{is=c}.

  To see \ref{mis=tc} we first consider an independent set $S$ of $G$ that
  is not maximal independent. Hence there is a vertex $v \in V$ such that
  $S\cup\{v\}$ is still independent. By property~\ref{is=c} the set
  $S^+\cup\{v\}$ is a chain, and hence $S^+$
  is also a chain, but not a tight one.

  Now let $S$ be a maximal independent set of $G$. So by property~\ref{is=c}
  the set $S^+$ is a chain of $(V^+,\prec)$. Since $S$ is a maximal
  independent set of $G$, for all $v \in V$ the set $S\cup\{v\}$ is not
  independent in $G$. By property~\ref{is=c} the set $S^+$
  is not a chain, and therefore $S^+$ is a tight chain of
  $(V^+,\prec)$.
\end{proof}

For every vertex $v \in V^+$ let $G_v = G[\{u \in V \mid u \preq v\}]$.
Especially, $G_{\bot} = (\es,\es)$ and $G_{\top} = G$.
Let $V_v = \{u \in V^+ \mid u \preq v\}$. For every $v \in V^+$, the relation
$\prec$ restricted to $V_v$ is a partial order with unique minimal element $\bot$
and unique maximal element $v$.

\subsection{Independent sets} \label{ss:is}

For every vertex $v \in V^+$ let $A(v)$ be the set of independent sets
of $S$ of $G_v$ with $v \in S$, and $a(v) = |A(v)|$.  That is,
$a(\top)$ is the number of independent sets of $G$. We can evaluate $A$
and $a$ recursively as follows:
\begin{align*}
  &A(\bot) = \{\es\} &&  a(\bot) = 1 \\
  &A(v_i)    = \bigcup_{u \prec v_i} \big\{S\cup\{v_i\} \mid S \in A(u)\big\} &&
  a(v_i)    = \sum_{u \prec v_i} a(u) &(i\in[1,n]) \\
  &A(\top) = \bigcup_{u \prec \top} A(u) && a(\top) = \sum_{u \prec \top} a(u)
\end{align*}
The sets $A(v)$ can be exponential in size (for example if $E=\es$), but the
recurrence for $a$ can be evaluated efficiently. We use the
linear extension of $\prec$ to evaluate the equations above in the order
$\bot,v_1,v_2,\ldots,v_2,\top$, as shown. At the end $a(\top)$ is the number
of independent sets in $G$.

We show that the recurrence above is correct by proving a sequence of lemmas,
which themselves are proven directly or by induction on $\prec$.

\begin{lemma} \label{l:is in=>is}
  For all $v \in V^+$ every set $S \in A(v)$ is independent in $G_v$
  and $v \in V$ implies $v \in S$.
\end{lemma}

\begin{proof}
  The base of the induction is for $v=\bot$.
  Clearly $\es$ is the unique independent set of $G_{\bot}$.

  Now let $v \in V$. Every set $S \in A(v)$ contains the vertex $v$
  and vertices in $A(u)$ for $u\prec v$. By induction hypothesis $S\sm\{v\}$
  is an independent set of some $G_u$. Therefore $S$ only contains vertices
  of $G_v$. Moreover $S\sm\{v\}$ is a chain in $(V,\prec)$. By transitivity,
  $S\sm\{v\} \in A(u)$ and $u \prec v$ imply that $S$ is also a chain in $(V,\prec)$
  and therefore an independent set of $G_v$.

  The latter argument also applies to $v=\top$.
\end{proof}

\begin{lemma} \label{l:is is=>in}
  Every nonempty independent set $S$ of $G$
  is contained in $A(\max(S))$.
\end{lemma}

\begin{proof}
  Let $S$ be an independent set of $G$. This implies that $S$ is a chain
  of $(V,\prec)$ and therefore $S$ has a unique maximal element $v$. We have
  $\es \in A(\bot)$ and, by induction hypothesis, $S\sm\{v\} \in A(u)$
  for some $u \prec v$. Consequently we have $S \in A(v)$, and clearly $v=\max(S)$.
\end{proof}

\begin{lemma} \label{l:is disjoint}
  For different $u,v \in V\cup\{\bot\}$ we have $A(u) \cap A(v) = \es$.
\end{lemma}

\begin{proof}
  If $u$ and $v$ are $\prec$-incomparable then these are adjacent vertices of $G$.
  Hence no independent set of $G$ can contain both $u$ and $v$.
  By Lemma~\ref{l:is in=>is} for all $S \in A(u)$ and all $T \in A(v)$
  we have $u \in S \sm T$ and $v \in T \sm S$, that is, $S \ne T$.
  Consequently $A(u) \cap A(v) = \es$ holds.

  Now let $u$ and $v$ be $\prec$-comparable.
  By symmetry we may assume $u\prec v$, especially $v \in V$.
  Again, Lemma~\ref{l:is is=>in} implies, for all $S \in A(u)$ and all
  $T \in A(v)$, that $v \in T \sm S$ holds. As before this means
  $S \ne T$, and therefore $A(u) \cap A(v) = \es$.
\end{proof}

We can implement the above recurrences for $a$ directly, but it is clear that the time
complexity is $O(n+\bar{m})$, not $O(n+m)$, since the summations are over the edges of
$\bar{E}$. If $\bar{m}<m$, this clearly implies $O(n+m)$ time, but otherwise we correct this as follows.

For all $v \in V^+$ we define $t(v)$ by
\begin{align*}
  t(\bot) &= 0 \\
  t(v_i)  &= a(\bot)+\sum_{j=1}^{i-1} a(v_j) && (i\in[1,n]) \\
  t(\top) &= a(\bot)+\sum_{j=1}^n a(v_j)
\end{align*}
The values of $t$ and $a$ are mutually recursive:
\begin{align*}
  t(\bot) &= 0 &  a(\bot) &= 1 \\
  t(v_1)  &= 1 &  a(v_1)  &= 2 \\
  t(v_i)  &= t(v_{i-1}) + a(v_{i-1}) &
  a(v_i)  &= t(v_i) - \hn\sum_{j<i,\,v_j \nprec v_i}\hn a(v_j) && (i\in[2,n]) \\
  t(\top) &= t(v_n) + a(v_n) & a(\top) &= t(\top)
\end{align*}
The recurrence for $t(v_i)$ is obvious, and for $a(v_i)$ we have
\[ t(v_i) - \hn\sum_{j<i,\,v_j \nprec v_i}\hn a(v_j)
=  a(\bot)+\sum_{j=1}^{i-1}a(v_j) - \hn\sum_{j<i,\,v_j \nprec v_i}\hn a(v_j)
=  \sum_{u \prec v_i} a(u)
=  a(v_i) \,, \]
as required. Now the summations are over the edges of $E$, and the algorithm is $O(n+m)$ time. Thus we have

\begin{theorem} \label{thm:is}
  Given a cocomparability graph on $n$ vertices, we can compute the number
  of its independent sets in time $O(n+m^*)$.
\end{theorem}

\begin{proof}
  Lemmas~\ref{l:is in=>is} and \ref{l:is is=>in} imply that, for all
  $v \in V^+$, the sets $A(v)$ defined by the recurrence contains
  all independent sets of $G_v$, especially $A(\top)$ is the set of all
  independent sets of $G$. It remains to show that $a(v) = |A(v)|$.
  To see this we observe that the unions over all $u\prec v$ are always disjoint
  by Lemma \ref{l:is disjoint}.

  Our algorithm just recurrences for $a(v)$ for all $v \in V^+$ in
  the order of linear extension of $\prec$. Then $a(\top)$ is the number
  of independent sets of $G$. Assuming additions can be performed in
  constant time, this takes $O(n+m^*)$ time, since each edge or non-edge
  of $G$ is used exactly once in the recurrences.
\end{proof}
The alternative evaluation of the recurrences which leads to $O(n+m^*)$
can be used in all the algorithms below, and we will not elaborate further
on this.

On the other hand, we have the following.

\begin{theorem} \label{thm:cliques}
  It is \emph{\#P}-complete to count cliques in a cocomparability graph.
\end{theorem}

\begin{proof}
Counting cliques in cocomparability graphs is equivalent to counting independent sets in comparabilty graphs.
Bipartite graphs are a subclass of comparability graphs. It is \#P-complete to count independent sets in
bipartite graphs~\cite{PrBa83,Vadhan}.
\end{proof}

In fact, even approximately counting cliques in cocomparability graphs appears to be hard, by the same argument,
being equivalent to the canonical problem \#BIS~\cite{DGGJ03}.

Finally, we note  that counting independent sets in cocomparability graphs is equivalent
to counting chains in partially ordered sets, and counting cliques is equivalent to counting
antichains. See~\cite{PrBa83} and~\cite{DGGJ03}, where antichains are called downsets. Thus
the results of this paper could be recast in  the language of partial orders.

\subsection{Maximal independent sets} \label{ss:mis}

Similarly we can compute the number of maximal independent sets. For every
vertex $v \in V^+$ let $B(v)$ denote the set of maximal independent sets
$S$ of $G_v$ with $v \in S$, and let $b(v) = |B(v)|$.
We can compute $B$ and $b$ as follows:
\begin{align*}
  B(\bot) &= \{\es\} &
  b(\bot) &= 1 \\
  B(v_i)    &= \bigcup_{u \scov v_i} \big\{S\cup\{v\} \mid S \in B(u)\big\} &
  b(v_i)    &= \sum_{u \scov v_i} b(u) && (i\in[1,n]) \\
  B(\top) &= \bigcup_{u \scov \top} B(u) &
  b(\top) &= \sum_{u \scov \top} b(u)
\end{align*}
The only difference to $A$ and $a$ is that the partial order $\prec$ is replaced
by its cover relation $\cov$, as anticipated in Lemma~\ref{l:chain}. For all
$u \prec v$ and all maximal independent sets $S'$ of $G_u$ the set $S'\cup\{v\}$
is maximal independent if and only if there is no $w \in V$ that is
$\prec$-between $u$ and $v$, hence if and only if $u \cov v$.

The recurrence for $B$ and $b$ can be shown to be correct by arguments similar
to the ones given above to justify the recurrence for $A$ and $a$.

\begin{lemma} \label{l:mis in=>mis}
  For all $v \in V^+$ every set $S \in B(v)$ is a maximal independent set of
  $G_v$ and $v \in V$ implies $v \in S$.
\end{lemma}

\begin{proof}
  The base of the induction is for $v=\bot$.
  Clearly $\es$ is the unique independent set of $G_{\bot}$.

  Now let $v \in V$. Every set $S \in B(v)$ contains the vertex $v$
  and vertices in $B(u)$ for $u \cov v$. By induction hypothesis $S\sm\{v\}$
  is a maximal independent set of some $G_u$. Therefore $S$ only contains
  vertices of $G_v$. Moreover $S\sm\{v\}$ is a tight chain in $(V_u,\prec)$.
  Since $u \cov v$ implies $u\prec v$ and because $\prec$ is transitive,
  $S\sm\{v\} \in B(u)$ and $u \cov v$ imply that $S$ is also a tight
  chain in $(V_v,\prec)$ and therefore a maximal independent set of $G_v$.

  A similar argument proves the assertion for $v=\top$.
\end{proof}

\begin{lemma} \label{l:mis mis=>in}
  For every $v \in V^+$ every maximal independent set $S$ of $G_v$
  satisfies $S \in B(v)$ and $v \in V$ implies $v \in S$.
\end{lemma}

\begin{proof}
  Let $S$ be a maximal independent set of $G_v$. By Lemma~\ref{l:chain},
  $S\cup\{\bot\}$ is a tight chain of $(V_v,\prec)$ with minimal element $\bot$
  and maximal element $v$. Therefore $v \in V$ implies $v \in S$.
  We show $S \in B(v)$ by induction. For $v=\bot$ we have $\es \in B(\bot)$.
  Otherwise, by induction hypothesis and the tightness stated above,
  $S\sm\{v\} \in B(u)$ for some $u \cov v$. Consequently we have $S \in B(v)$.
\end{proof}

\begin{theorem} \label{thm:mis}
  Given a cocomparability graph on $n$ vertices, we can compute the number
  of its maximal independent sets in $O(n+m^*)$ time.
\end{theorem}

\begin{proof}
  Lemmas~\ref{l:mis in=>mis} and \ref{l:mis mis=>in} imply that, for all
  $v \in V^+$, the sets $B(v)$ defined by the recurrence contains indeed
  all maximal independent sets of $G_v$, especially $B(\top)$ is the set
  of all maximal independent sets of $G$.

  It remains to show that $b(v) = |B(v)|$. Again, the unions over all
  $u \cov v$ are always disjoint because $B(u) \subseteq A(u)$ holds for all
  $u$, and by Lemma \ref{l:is disjoint}.

  The time analysis is essentially the same as in the proof of Theorem~\ref{thm:is}.
\end{proof}
Corresponding to Theorem~\ref{thm:cliques}, we also have the following
\begin{theorem} \label{thm:maxcliques}
  It is \emph{\#P}-complete to count maximal cliques in a cocomparability graph.
\end{theorem}

\begin{proof}
By the same argument as Theorem~\ref{thm:cliques}, using that it is \#P-complete to count
maximal independent sets in bipartite graphs~\cite{maxis}.
\end{proof}

\subsection{Independent sets of size \boldmath$k$} \label{ss:isk}

Next we consider independent sets of size exactly $k$ for some fixed value of
$k \in [0,n]$. For every vertex $v \in V^+$ and every integer $i \in [0,k]$,
the set $C(v,i)$ of independent sets $S$ of $G_v$ with $v \in S$ and $|S|=i$,
and size $c(v,i)$, satisfy the following recurrences:
\begin{align*}
  C(\bot,0) &= \{\es\} &\quad
  c(\bot,0) &= 1 \\
  C(\bot,i) &= \es &
  c(\bot,i) &= 0   && (i \in [1,k])\\
  C(v_j,0)    &= \es &
  c(v_j,0)    &= 0   && (j \in [1,n]) \\
  C(v_j,i)    &= \bigcup_{u \prec v_j} \big\{S\cup\{v_j\} \mid S \in C(u,i-1)\big\} &
  c(v_j,i)    &= \sum_{u \prec v_j} c(u,i-1)
            && (i \in [1,k],\ j \in [1,n]) \\
  C(\top,i) &= \bigcup_{u \prec \top} C(u,i) &
  c(\top,i) &= \sum_{u \prec \top} c(u,i) && (i \in [0,k])
\end{align*}
The correctness of these recurrences is again based on the fact that,
for every $v \in V$ and every $i \in [1,k]$, every independent set
$S$ of $G_v$ with $v \in S$ and $|S|=i$ there is an independent set
$S'$ of $G_u$ of size $i-1$ for some $u \prec v$, where $u=\bot$ if $S'=\es$
and otherwise $u$ is the $\prec$-maximal vertex in $S'$. For $v=\top$ we
have $|S|=|S'|$ because $\top \notin V$. Then we have the following.

\begin{lemma} \label{l:isk in=>is}
  For all $v \in V$ and all $i\in[1,k]$ every set $S \in C(v,i)$ is a
  independent set of $G_v$ with $|S|=i$ and $v \in S$.
\end{lemma}

\begin{lemma} \label{l:isk is=>in}
  For every $i \in [1,k]$ every independent set $S$ of size $i$
  is contained in $C(\max(S),i)$.
\end{lemma}

\begin{theorem} \label{thm:isk}
  Given a cocomparability graph on $n$ vertices and a number $k \in [0,n]$, we
  can compute the number of its independent sets of size exactly $k$ in time
  $O(k(n+m^*))$. In time $O(n^2+nm^*)$ we can do this for all $k \in [0,n]$.
\end{theorem}

\begin{proof}
  Lemmas~\ref{l:is in=>is} and \ref{l:is is=>in} imply that, for all
  $v \in V^+$, the sets $A(v)$ defined by the recurrence contains indeed
  all independent sets of $G_v$, especially $A(\top)$ is the set of all
  independent sets of $G$. It remains to show that $a(v) = |A(v)|$.
  To see this we observe that the unions over all $u\prec v$ are always disjoint
  by Lemma \ref{l:is disjoint}.

  Our algorithm just evaluates the sums for $a(v)$ for all $v \in V^+$ in
  an order that is a linear extension of $\prec$. Then $a(\top)$ is the number
  of independent sets of $G$. Assuming additions can be performed in
  constant time, this takes $O(n^2)$ time.
\end{proof}

As a by-product the algorithm of Theorem~\ref{thm:isk} computes the size
$\alpha(G)$ of a maximum independent set of $G$, which is
$\max\{i \mid c(\top,i)>0\}$, and the number of maximum independent set
in $G$, which is $c(\top,\alpha(G))$. Since every maximum independent set
is also maximal, the algorithm from Theorem~\ref{thm:misk} can be used as
well and should be faster on average.

Once the number $c(\top,i)$ of independent sets of size exactly $i$ has been
determined for all $i \in [0,\alpha(G)]$ we can evaluate the independent set
polynomial $\sum_{i=0}^{\alpha(G)} c(\top,i)x^i$ for all values of $x$.

We also have the following fixed-parameter hardness result for counting $k$-cliques in cocomparability graphs.
\begin{theorem} \label{thm:kcliques}
  It is \emph{\#W[1]}-complete to count $k$-cliques in a cocomparability graph.
\end{theorem}

\begin{proof}
By the same argument as Theorem~\ref{thm:cliques}, using that it is \#W[1]-complete to count
independent sets of size $k$ in bipartite graphs~\cite[Theorem 4]{CDFGL}.
\end{proof}

\subsection{Maximal independent sets of size \boldmath$k$} \label{ss:misk}

For every vertex $v \in V^+$ and every integer $i \in [0,k]$, let
$D(v,i)$ be the set of maximal independent sets $S$ of $G_v$
with $v \in S$ and $|S|=i$, and $d(v,i) = |D(v,i)|$. We have
\begin{align*}
  D(\bot,0) &= \{\es\} &
  d(\bot,0) &= 1 \\
  D(\bot,i) &= \es &
  d(\bot,i) &= 0   && (i \in [1,k]) \\
  D(v_j,0)    &= \es &
  d(v_j,0)    &= 0   && (j \in [1,n] )\\
  D(v,i)    &= \bigcup_{u \scov v} \big\{S\cup\{v\} \mid S \in D(u,i-1)\big\} &
  d(v,i)    &= \sum_{u \scov v} d(u,i-1)
            && (i \in [1,k], j \in [1,n]) \\
  D(\top,i) &= \bigcup_{u \scov \top} D(u,i) &
  d(\top,i) &= \sum_{u \scov \top} d(u,i) && (i \in [0,k])
\end{align*}
The following are proved similarly to the correponding results above.
\begin{lemma} \label{l:misk in=>mis}
  For all $v \in V$ and all $i\in[1,k]$ every set $S \in D(v,i)$ is a
  maximal independent set of $G_v$ with $|S|=i$ and $v \in S$.
\end{lemma}

\begin{lemma} \label{l:misk mis=>in}
  For every $v \in V^+$ and $i \in [1,k]$ every maximal independent set $S$ of
  size $i$ in $G_v$
  is contained in $D(v,i)$, and $v \in V$ implies $v \in S$.
\end{lemma}

\begin{theorem} \label{thm:misk}
  Given a cocomparability graph on $n$ vertices and a number $k \in [0,n]$, we
  can compute the number of its maximal independent sets of size exactly $k$
  in time $O(k(n+m^*))$. In time $O(n^2+nm^*)$ we can do this for all $k \in [0,n]$.
\end{theorem}

We have no corresponding hardness result in this case, since the complexity of
counting \emph{maximal} $k$-independent sets in bipartite graphs appears to be open.

\section{Comparability graphs}\label{sec:comp}

Counting independent sets in a cocomparability graph $G$ is equivalent to counting cliques
in its complement $\bar{G}$. Since our algorithms are symmetrical between $G$ and $\bar{G}$,
all the results of section~\ref{sec:cocomp} remain true by interchanging the words
``cocomparability'' and ``cocomparability'', and the words ``clique'' and ``independent set''.
Therefore, we will not detail the modified results.

\section{Permutation graphs}\label{sec:perm}

Permutation graphs are both comparability and cocomparability graphs, so the algorithms
of section~\ref{sec:cocomp} are valid both for counting independent sets in permutation graphs
and (with obvious modifications) for counting cliques in permutation graphs. The result
of~\cite{bip perm} for counting independent sets in bipartite permutation graphs is a special case.

\end{document}